\documentclass[11pt,a4paper]{article}
\evensidemargin=\oddsidemargin
\def\squarebox#1{\hbox to #1{\hfill\vbox to #1{\vfill}}}
\newcommand{\qed}{\hspace*{\fill}
            \vbox{\hrule\hbox{\vrule\squarebox{.667em}\vrule}\hrule}\smallskip\newline}
\newtheorem{THEOREM}{Theorem}
\newenvironment{theorem}{\begin{THEOREM} \hspace{-.85em} {\bf :} \rm}                        {\end{THEOREM}}
\newtheorem{LEMMA}[THEOREM]{Lemma}
\newenvironment{lemma}{\begin{LEMMA} \hspace{-.85em} {\bf :} \rm}                      {\end{LEMMA}}
\newtheorem{COROLLARY}[THEOREM]{Corollary}
\newenvironment{corollary}{\begin{COROLLARY} \hspace{-.65em} {\bf :} \rm}                          {\end{COROLLARY}}
\newtheorem{FACT}[THEOREM]{Fact}

\newenvironment{proof}{\noindent {\bf Proof:} \hspace{-.2em}} {}

\newtheorem{DEFINITION}{Definition}
\newenvironment{definition}{\begin{DEFINITION} \hspace{-.85em} {\bf :} \rm}
                            {\end{DEFINITION}}
\newtheorem{PROPOSITION}{Proposition}

\newtheorem{CLAIM}[THEOREM]{Claim}

\usepackage{graphicx}

\usepackage{subfigure}
\usepackage{graphicx}
\usepackage{graphics}
\usepackage{amssymb}
\input{epsf}

\usepackage{algorithm}
\usepackage{algpseudocode}
\begin{document}
\title{Parameterized algorithms for the 2-clustering problem with minimum sum and minimum sum of squares objective functions}
\author{Bang Ye Wu\thanks{National Chung Cheng University, ChiaYi, Taiwan 621,
R.O.C., E-mail: bangye@cs.ccu.edu.tw}
 and Li-Hsuan Chen
\\
{\small Dept. of Computer Science and Information Engineering}\\
{\small National Chung Cheng University, Taiwan}}
\date{}
\maketitle
\subsection*{\centering Abstract}
{\em
In the {\sc Min-Sum 2-Clustering} problem, we are given a graph and a parameter $k$, and the goal is to determine if there exists a 2-partition of the vertex set such that the total conflict number is at most $k$, where the conflict number of a vertex is the number of its non-neighbors in the same cluster and neighbors in the different cluster. 
The problem is equivalent to {\sc 2-Cluster Editing} and {\sc 2-Correlation Clustering} with an additional multiplicative factor two in the cost function. 
In this paper we show an algorithm for {\sc Min-Sum 2-Clustering} with time complexity
$O(n\cdot 2.619^{r/(1-4r/n)}+n^3)$, where $n$ is the number of vertices and $r=k/n$.
Particularly, the time complexity is $O^*(2.619^{k/n})$ for $k\in o(n^2)$ and polynomial for $k\in O(n\log n)$, which implies that the problem can be solved in subexponential time for $k\in o(n^2)$.
We also design a parameterized algorithm for a variant in which the cost is the sum of the squared conflict-numbers. For $k\in o(n^3)$, the algorithm runs in subexponential $O(n^3\cdot 5.171^{\theta})$ time, where $\theta=\sqrt{k/n}$.  }
{\flushleft\bf Key words.} parameterized algorithm, kernelization, cluster graph, clustering, graph modification.
{\flushleft\bf AMS subject classifications.} 65W05, 68R10, 68Q25, 05C85, 91C20

\section{Introduction}

{\flushleft \em Problem definition and motivation.}
Clustering is an important concept with applications in numerous fields, and the problem {\sc Cluster Editing}, also known as {\sc Correlation Clustering}, 
is a graph theoretic approach to clustering \cite{Bansal,sha04}.
A \emph{cluster graph} is a graph whose every connected component is a clique.
{\sc Cluster Editing} asks for the minimum number of edge insertions and  deletions to modify the input graph into a cluster graph. 
One of the studied variants is {\sc $p$-Cluster Editing}, in which one is asked to modify the input graph into exactly $p$ disjoint maximal cliques \cite{sha04}.
In many applications, graph edges represent the similarities between items (vertices), and one wants to partition items into clusters such that items in the same cluster are similar and items in different clusters are dissimilar. Hence, an edited edge can be thought of as a \emph{conflict} (or \emph{disagreement} as named in {\sc Correlation Clustering}) between two items in the clustering.  
Let the \emph{conflict number} of a vertex be the number of edited edges incident to it. Then, feasible sets of edits of cardinality $k$ correspond one-to-one to clusterings with total conflict number $2k$. That is, {\sc Cluster Editing} is equivalent to finding a vertex partition with minimum total conflict number.
The transformation of the problem definition provides us easier ways to define other meaningful objective functions on the conflict numbers, such as the maximum conflict number and the sum of squared conflict numbers.

In this paper we focus on 2-clusterings which partition the vertices into two subsets called clusters.
For an input graph $G=(V,E)$ and a 2-partition $\pi=(V_1,V_2)$ of $V$, two vertices $u$ and $v$ are in \emph{conflict} if they are in the same cluster but $(u,v)\notin E$ or they are in different clusters but $(u,v)\in E$.
Let $c_\pi(v)$ denote the number of vertices in conflict with $v$ in $\pi$.
A 2-clustering problem in general asks for a 2-partition such that the conflict numbers are as ``small'' as possible. As in many optimization problems, there are different ways to define ``small'' for a set of quantities (or a vector). Possibly the most frequently used cost functions are min-sum, min-sum of squares, and min-max, which are equivalent to the 1-norm, 2-norm, and $\infty$-norm for defining the length of a vector in linear algebra, respectively. Let $h_1(\pi)=\sum_{v\in V}c_\pi(v)$ and  $h_2(\pi)=\sum_{v\in V}c^2_\pi(v)$ be costs of a 2-partition $\pi$. 
The two problems studied in this paper are formally defined as follows.
We focus on their decision versions.
\begin{quote}
{\sc Problem:} {\sc Min-Sum 2-Clustering}\\
{\sc Instance}: A graph $G=(V,E)$ and a nonnegative integer $k$.\\
{\sc Question}: Is there a 2-partition $\pi=(V_1,V_2)$ of $V$ such that $h_1(\pi)\leq k$ and $V_1,V_2\neq\emptyset$?
\end{quote}
Following the definition in the literature, we exclude the case that $V_1$ or $V_2$ is empty. 
The second problem, named {\sc Min-Square 2-Clustering}, is defined similarly except the cost is defined by $h_2$. 
Intuitively, while {\sc Min-Sum 2-Clustering} seeks to minimize the total conflict number, {\sc Min-Square 2-Clustering} looks for 2-partitions simultaneously minimizing the total and the individual conflict numbers because $h_2(\pi)=\sum_{v\in V}c^2_\pi(v)=h_1^2(\pi)/n+n\sigma^2$, where $n$ is the number of vertices and $\sigma^2$ is the \emph{variance} of the conflict numbers.
  
{\flushleft \em Previous results.}
Shamir et al. \cite{sha04} studied the computational complexities of three edge modification problems. {\sc Cluster Editing} asks for the minimum total number of edge insertions and deletions to modify a graph into a cluster graph, while in {\sc Cluster Deletion} (respectively, {\sc Cluster Completion}), only edge deletions (respectively, insertions) are allowed.
They showed that {\sc Cluster Editing} is NP-hard, {\sc Cluster Deletion} is Max SNP-hard, and {\sc Cluster Completion} is polynomial-time solvable.
They also showed that {\sc $p$-Cluster Deletion} is NP-hard for any $p>2$ but polynomial-time solvable for $p=2$, and {\sc $p$-Cluster Editing} is NP-hard for any $p\geq 2$.
The parameterized version of {\sc Cluster Editing} and variants of it were studied intensively \cite{fp1,boc11,fp2,fp3,fp4,fp5,fp6,fp7}. A variant with vertex (rather than edge) deletions was considered in \cite{huf07}, and another variant in which overlapping clusters are allowed was studied in \cite{fel11}.

{\sc Cluster Editing} is equivalent to {\sc Correlation Clustering} on complete signed graphs.
In a signed graph, each edge is labelled by ``+'' or ``-'', representing that the two items are similar or dissimilar; or the two persons like or dislike each other in social network analysis.
For a clustering, a positive edge within a cluster or a negative edge between clusters is an \emph{agreement}, and a positive edge across clusters or a negative edge inside a cluster is a \emph{disagreement}. 
The maximization version of the {\sc Correlation Clustering} problem 
seeks to maximize the number of agreements, while the minimization version aims to minimize the number of disagreements.
{\sc Correlation Clustering} on complete signed graphs was formulated and studied in \cite{Bansal}, in which the authors presented a PTAS for the maximization version and a constant factor approximation algorithm for the minimization version.
In \cite{ail08}, Ailon et al. proposed a simple randomized algorithm for the minimization version. For the unweighted case, the expected approximation ratio is three. They also showed that it is a 5-approximation algorithm for the weighted case with the \emph{probability constraints}, i.e., it is assumed that  $w_{ij}^-+w_{ij}^+=1$ for each pair $(i,j)$ of vertices, where $w_{ij}^+$ and $w_{ij}^-$ are the nonnegative weights of the ``+'' and ``-'' edges between $i$ and $j$, respectively. If in addition the weights satisfy the triangle inequality ($w_{ik}^-\leq w_{ij}^-+w_{jk}^-$ for all vertices $i,j,k$), then the approximation ratio is two.

For a constant $p\geq 2$, {\sc $p$-Correlation Clustering} is a variant of {\sc Correlation Clustering} such that the vertices are partitioned into exactly $p$ clusters. 
While the minimization version of {\sc Correlation Clustering} on complete signed graphs is APX-hard \cite{Cha05}, Giotis and Guruswami showed that both the minimization and the maximization versions of {\sc $p$-Correlation Clustering} admit PTAS for unweighted complete signed graphs \cite{Giotis}.
{\sc 2-Correlation Clustering} is also known as {\sc Balanced Subgraph} which name comes from the application in social network analysis \cite{Har,huf07b,snabook}.
Another related problem studied in the literature is {\sc Consensus Clustering}~\cite{ail08,bon09,bon08,fil03}.

For {\sc $p$-Cluster Editing}, a kernel with $(p+2)k+p$ vertices was given by Guo \cite{fp7}. A variant such that the conflict number of each vertex must be bounded by a parameter was studied in \cite{kom12}. The problem of finding a 2-clustering minimizing the maximum conflict number is NP-hard \cite{mm2ej}.
Very recently, Fomin et al. gave a parameterized algorithm with time complexity $O(2^{O(\sqrt{pk})}+n^2)$, where $n$ is the number of vertices \cite{fomin13}. They also showed a lower bound for the parameterized complexity: For any constant $0\leq \sigma\leq 1$ there exists a function $p(k)\in \Theta(k^\sigma)$ such that {\sc $p$-Cluster Editing} restricted to instances with $p=p(k)$ cannot be solved in $2^{o(\sqrt{pk})}\cdot n^{O(1)}$ time unless the Exponential Time Hypothesis fails.

{\flushleft \em Our contributions.}
We develop parameterized algorithms for {\sc Min-Sum 2-Clustering} and {\sc Min-Square 2-Clustering}.
For {\sc Min-Sum 2-Clustering}, the algorithm runs in  $O(n\cdot \phi^{2r/(1-4r/n)}+n^3)$ time, where $\phi\approx 1.618$ and $r=k/n$.
When $k\in o(n^2)$, the time complexity is $O(n\cdot 2.619^{k/n}+n^3)$. Our result implies that the problem, as well as {\sc 2-Cluster Editing} can be solved in subexponential, i.e. $O^*(2^{o(n)})$, time for $k\in o(n^2)$, where the $O^*(\cdot)$ notation ignores factors polynomial in $n$.
In particular it is polynomial-time solvable for $k\in O(n\log n)$. We also note that the time complexity is better than $O^*(2^{O(\sqrt{pk})})$ recently obtained by Fomin et al. \cite{fomin13} for the special case of $p=2$.
Even when $k=\delta n^2$ with small constant $\delta$, our algorithm improves the brute-force algorithm significantly. For example, when $\delta=0.1$, we have that $r=k/n=0.1n$, and the time complexity is $O^*(\phi^{2r/(1-4r/n)})\approx O^*(1.174^{n})$, much better than $O^*(2^n)$.

For {\sc Min-Square 2-Clustering} with cost bound $k\in o(n^3)$, our algorithm runs in subexponential $O(n^3\cdot 5.171^{\sqrt{k/n}})$ time, which also implies that the problem is polynomial-time solvable for $k\in O( n\log^2 n)$.

Both of the algorithms look for a 2-partition and the main steps are sketched as follows.
\begin{enumerate}
\item Guess a vertex that has the smallest conflict number in an optimal 2-partition and set an initial 2-partition according to its neighborhood. 
\item Using kernelization-style rules, determine for most of the vertices whether they should be swapped to the other cluster, and leave a set of undetermined vertices of size $O(k/n)$ in the case of {\sc Min-Sum 2-Clustering}, or $O(\sqrt{k/n})$ in the case of {\sc Min-Square 2-Clustering}.
\item Apply a standard branching algorithm to the undetermined vertices.
\end{enumerate}

{\flushleft \em Organization of the paper.}
In Section 2 we give some notation and definitions, as well as some properties used in this paper. The reduction algorithm is in Section~3. In Sections 4 and 5 we show the algorithms for the two problems, respectively. Finally some concluding remarks are in Section~6.

\section{Preliminaries}

An instance of a parameterized problem consists of $(I,k)$, where $k$ is the parameter. A problem is \emph{fixed-parameter tractable} (FPT) if it can be solved in time complexity $O(f(k)\cdot q(|I|))$, where $f$ is an arbitrary computable function of $k$ and $q$ is a polynomial in the input size.
For more details about parameterized complexity, we refer to the book of Downey and Fellows \cite{fptbook}.
\emph{Kernelization} is a widely-used technique for parameterized algorithms.
In polynomial time, a kernelization algorithm converts an instance $(I,k)$ to a reduced instance $(I',k')$, called a \emph{kernel} such that the answer is not changed, $k'\leq k$ and $|I'|$ is bounded by a computable function of $k$.

For two sets $S_1$ and $S_2$, the set difference is denoted by $S_1\setminus S_2$, and
the symmetric difference is denoted by $S_1\ominus S_2=(S_1\setminus S_2)\cup (S_2\setminus S_1)$. For simplicity, $S_1\ominus v=S_1\ominus \{v\}$.
Throughout this paper, $G=(V,E)$ is the input graph and $n=|V|$.
For a vertex $v$, let $N_G[v]=\{u\mid (u,v)\in E\}\cup \{v\}$ denote the closed neighborhood of $v$ in $G$. 

For a vertex set $V$, a \emph{2-partition} of $V$ is an unordered pair $\pi=(V_1,V_2)$ of subsets of $V$ such that $V_1\cap V_2=\emptyset$ and $V_1\cup V_2=V$. The two subsets $V_1$ and $V_2$ are called \emph{clusters}. Two vertices $u$ and $v$ are in \emph{conflict} if they are in the same cluster but $(u,v)\notin E$ or they are in different clusters but $(u,v)\in E$. 
Let $C_\pi(v)$ denote the set of vertices in conflict with $v$ in $\pi$ and $c_\pi(v)=|C_\pi(v)|$ be the \emph{conflict number} of $v$. We assume that $v\notin C_\pi(v)$ for each vertex $v$. 
When there is no confusion, we shall omit the subscript and simply use $c(\cdot)$ instead of $c_\pi(\cdot)$.
For $u,v\in V$ and $S\subseteq V$, let $C_\pi(v,S)=C_\pi(v)\cap S$, $c_\pi(v,S)=|C_\pi(v,S)|$, and $c_\pi(v,u)=c_\pi(v,\{u\})$. 
For two vertex subsets $S_1$ and $S_2$, let $c_\pi(S_1,S_2)=\sum_{v\in S_1}c_\pi(v,S_2)$. Note that $c_\pi(S_1,S_2)=c_\pi(S_2,S_1)$.
When the 2-partition $\pi$ is clear from the context, we shall also omit the subscript in  $c_\pi(\cdot,\cdot)$ and $C_\pi(\cdot,\cdot)$.

A graph $G=(V,E)$ is a \emph{2-cluster graph} if it consists of exactly two disjoint maximal cliques. 
In the literature, {\sc Min-Sum 2-Clustering} is also known as {\sc 2-Cluster Editing}.
For $G=(V,E)$, a set $D\subseteq V\times V$ is a \emph{set of edits} (to 2-cluster graph) for $G$ if $G'=(V,E\ominus D)$ is a 2-cluster graph. In other words, $G$ can be modified into a 2-cluster graph by inserting $D\setminus E$ and deleting $D\cap E$.
Given a graph $G$ and integer $k$, {\sc 2-Cluster Editing} asks if there is a set of edits $D$ for $G$ such that $|D|\leq k$.
Let $E^*(\pi)=\bigcup_{i=1,2}\{(u,v)\mid u,v\in V_i\}$ which is the edge set of the 2-cluster graph. By definition, if $D$ is a set of edits and $\pi$ is the corresponding 2-partition, then $E\ominus D=E^*(\pi)$ and $D=\bigcup_{v\in V}\{(v,u)\mid u\in C_\pi(v)\}$. Therefore, finding a set of edits is equivalent to finding the corresponding 2-partition. Furthermore, $|D|=(1/2)\sum_{v\in V}c_\pi(v)$. Consequently {\sc Min-Sum 2-Clustering} is equivalent to {\sc 2-Cluster Editing} with an additional multiplicative factor two in the cost function.

Both {\sc Min-Sum 2-Clustering} and {\sc Min-Square 2-Clustering} look for a 2-partition without empty clusters. The cost functions of the two problems are $h_1(\pi)\equiv \sum_v c_\pi(v)$ and $h_2(\pi)\equiv \sum_v c_\pi^2(v)$, respectively. Given a graph $G$ and a 2-partition $\pi$, computing $C_\pi(v)$ for all $v\in V$, as well as $h_1(\pi)$ and $h_2(\pi)$, can be easily done in $O(n^2)$ time by checking all pairs of vertices. 

To \emph{flip} a vertex $v$ in a 2-partition $\pi=(V_1,V_2)$ is to move $v$ to the other cluster, that is, we change $\pi$ to $\pi\ominus v\equiv (V_1\ominus v,V_2\ominus v)$.
Flipping a subset $S$ of vertices changes $(V_1,V_2)$ to $(V_1\ominus S,V_2\ominus S)$.
Suppose that $\pi'=\pi\ominus v$. Then $u\in C_{\pi'}(v)$ if and only if $u\notin  C_{\pi}(v)$ for all $u\neq v$. That is, flipping a vertex exchanges its  conflicting-relations with all the other vertices. Furthermore, when flipping a set $F$,
only those conflicting pairs in $F\times \bar{F}$ change, where $\bar{F}=V\setminus F$.
If $\pi'=\pi\ominus F$, then
\begin{eqnarray}
C_{\pi'}(v)=\left\{
\begin{array}{ll}
C_{\pi}(v,F)\cup (\bar{F}\setminus C_{\pi}(v,\bar{F}))&\mbox{if }v\in F;\\
C_{\pi}(v,\bar{F})\cup (F \setminus C_{\pi}(v,F))&\mbox{if }v\in \bar{F}.
\end{array}
\right.\label{eq:flip}
\end{eqnarray}

The \emph{profit} of a flipping set $F$, denoted by $\Delta(F)$, is the decrement of total conflict number after flipping $F$, i.e., $\Delta(F)=h_1(\pi)-h_1(\pi\ominus F)=\sum_v c_\pi(v)-\sum_v c_{\pi\ominus F}(v)$.
\begin{lemma}\label{lem:flip}
$\Delta(F)=4c_\pi(F,\bar{F})-2|F||\bar{F}|$.
\end{lemma}
\begin{proof}
By (\ref{eq:flip}), we only need to count the conflict pairs crossing $(F,\bar{F})$. Since $c_\pi(F,\bar{F})=c_\pi(\bar{F},F)$, we have
\begin{eqnarray*}
\Delta(F)=2(c_\pi(F,\bar{F})-(|F||\bar{F}|-c_\pi(F,\bar{F})))=4c_\pi(F,\bar{F})-2|F||\bar{F}|.
\end{eqnarray*}
\qed\end{proof}
\begin{corollary}\label{flipone}
$\Delta(\{v\})=4c_\pi(v)-2(n-1)$.
\end{corollary}

\section{Reduction algorithm}

We shall solve {\sc Min-Sum 2-Clustering} and {\sc Min-Square 2-Clustering} by finding flipping sets.
In this section we propose a reduction algorithm which reduces the search space of flipping sets, and the reduction will be used in the next two sections.

\begin{definition}
Let $\pi$ be a 2-partition of $V$ and $K,t$ be nonnegative integers. A vertex subset $F$ is a $(K,t)$-feasible flipping set for $\pi$ if $\sum_v c_{\pi'}(v)\leq K$ and $c_{\pi'}(v)\leq t$ for any vertex $v\in V$, where $\pi'=\pi\ominus F$.
When $K$ and $t$ are clear from the context, we shall simply say that $F$ is a feasible flipping set.
\end{definition}
By definition, an empty set may be a $(K,t)$-feasible flipping set.
The goal of this section is a reduction algorithm achieving the next lemma.
\begin{lemma}\label{lem:red}
Let $\pi_0$ be a 2-partition of $V$ and $K$, $t$, $f$ be nonnegative integers satisfying $t+2f<n$.
Given $(G,\pi_0,K,t,f)$, one can in $O(n^2)$ time compute a vertex subset $U$, a 2-partition $\pi$, and an integer $m$ such that
\begin{quote}
\begin{itemize}
	\item[$(i)$] $m\leq f$;  
	\item[$(ii)$] $|U|\leq K/(n-t-2m)-m$ or $m=0$; and 
	\item[$(iii)$] if there exists a $(K,t)$-feasible flipping set for $\pi_0$ of size at most $f$, then there exists a $(K,t)$-feasible flipping set $F\subseteq U$ for $\pi$ of size at most $m$.
\end{itemize}
\end{quote}
\end{lemma}

The bound $f$ on the size of flipping sets will be called  \emph{flipping quota}. 
In the remaining paragraphs of this section, we first describe the reduction algorithm, and then show the correctness of the reduction rules in  Lemmas~\ref{red-rule} and \ref{lem:ub}. The proof of Lemma~\ref{lem:red} will be delayed to the end of this section.
The reduction algorithm is based on the following reduction rules for any 2-partition $\pi$ and integer $f$. We omit the parameters $(K,t)$ in the rules.
\begin{quote}
\begin{itemize}
\item[R1:] If $c_\pi(v)>t+f$, then $v$ must be in any feasible flipping set for $\pi$ of size at most $f$.
\item[R2:] If $c_\pi (v)<n-t-f$, then $v$ cannot be in any feasible flipping set for $\pi$ of size at most $f$.
\item[R3:] If $c_\pi(v)\leq t+f$ for all $v\in V$ and $|U|> K/(n-t-2f)-f$, where $U=\{v\mid c_\pi(v)\geq n-t-f\}$, then there is no feasible flipping set for $\pi$ of size exactly $f$.
\end{itemize}
\end{quote}
Algorithm~\ref{alg:red} is the reduction algorithm.
We note that the upper bound of $|U|$ in R3 is for the case of flipping set of size \emph{exactly} $f$. When $|U|> K/(n-t-2f)-f$, it is still possible that there is a feasible flipping set of size less than $f$. Therefore, the algorithm iteratively decreases $f$ until the bound is satisfied or the flipping quota is zero.

\begin{algorithm}[t]
\caption{: {\sc Reduction}$(G,\pi_0,K,t,f)$}\label{alg:red}
{\bf Input:} a graph $G=(V,E)$, a 2-partition $\pi_0$, and integers $K$, $t$ and $f$.\\
{\bf Output:} a 2-partition $\pi$, a vertex subset $U$, and an integer $m$.
\begin{algorithmic}[1]
\State initially $\pi=\pi_0$;
\State compute $c_\pi(u)$ for each $u$ and construct $U\leftarrow\{u\mid c_\pi(u)\geq n-t-f\}$; 
\While{$\exists v\in U$ such that $c_\pi(v)>t+f$}\label{loop}
\State $\pi\leftarrow \pi\ominus v$; \Comment{flipping $v$}
\State remove $v$ from $U$ and update $c_\pi(u)$ for each $u$;
\State $f\leftarrow f-1$ and $U\leftarrow \{u\mid c_\pi(u)\geq n-t-f\}$;
\State {\bf if} $f=0$ {\bf then} goto step \ref{step:end};
\EndWhile
\If{$|U|>\frac{K}{n-t-2f}-f$}
\State $f\leftarrow f-1$ and $U\leftarrow \{u\mid c_\pi(u)\geq n-t-f\}$;
\State {\bf if} $f=0$ {\bf then} goto step \ref{step:end};
\State goto step \ref{loop};
\EndIf
\State \Return $(U,\pi,m=f)$;\label{step:end}
\end{algorithmic}
\end{algorithm}

\begin{lemma}\label{red-rule}
The reduction rules R1 and R2 are correct.
\end{lemma}
\begin{proof}
Since the conflict number of $v$ is decreased by at most one when another vertex is flipped, if $c_\pi(v)>t+f$ and $v$ is not flipped, then its conflict number will be larger than $t$. So R1 is correct.

For R2, if $c_\pi (v)<n-t-f$, then flipping $v$ will change its conflict number to $n-1-c_\pi (v)>t+f-1$, and further flipping $f-1$ vertices make the conflict number at least $t+1$. Therefore R2 is correct. 
\qed\end{proof}

We now show the upper bound of $|U|$ in rule R3.
\begin{lemma}\label{lem:ub}
Suppose that $t+2f<n$ and $c_\pi(v)\leq t+f$ for all $v\in V$.
If there exists a $(K,t)$-feasible flipping set $F\subseteq U$ for $\pi$ of size $f$, then $|U|\leq K/(n-t-2f)-f$, where $U=\{v\mid c_\pi(v)\geq n-t-f\}$.
\end{lemma}
\begin{proof}
Let $\bar{F}=V\setminus F$ and $X=V\setminus U$. We omit the subscript $\pi$ in the proof.
By Lemma~\ref{lem:flip}, 
\begin{eqnarray}
\sum_{v\in V}c(v)\leq K+4c(F,\bar{F})-2f(n-f). \label{ub1}
\end{eqnarray}
Let $Y=U\setminus F$. Since $V=Y\cup F\cup X$ and $Y, F, X$ are mutually disjoint, we have that
\[ \sum_{v\in V}c(v)=\sum_{v\in Y}c(v)+\sum_{v\in F}c(v)+\sum_{v\in X}c(v). \]
For $v\in F$, since $c(v)\geq c(v,\bar{F})$,
\begin{eqnarray}
\sum_{v\in F}c(v)\geq c(F,\bar{F}). \label{ub2}
\end{eqnarray}
Since $(X,Y)$ is a 2-partition of $\bar{F}$, we have that
\begin{eqnarray*}
c(X,F)&=&c(F,X)=\sum_{v\in F}c(v,X)\\
&=& \sum_{v\in F}(c(v,\bar{F})-c(v,Y))\\
&\geq& \sum_{v\in F}(c(v,\bar{F})-|Y|)\\
&=&c(F,\bar{F})-f|Y|,
\end{eqnarray*}
and then 
\begin{eqnarray}
\sum_{v\in X}c(v)\geq c(X,F)\geq c(F,\bar{F})-f|Y|. \label{ub3}
\end{eqnarray}
Therefore, by (\ref{ub1}), (\ref{ub2}) and (\ref{ub3}),
\begin{eqnarray*}
\sum_{v\in Y}c(v)&=&\sum_{v\in V}c(v)-\sum_{v\in F}c(v)-\sum_{v\in X}c(v) \\
&\leq& (K+4c(F,\bar{F})-2f(n-f))-c(F,\bar{F})-(c(F,\bar{F})-f|Y|)\\
&=&K+2c(F,\bar{F})-2f(n-f)+f|Y|.
\end{eqnarray*}
Since $c(v)\geq n-t-f$ for any $v\in Y$, we have that
\begin{eqnarray*}
(n-t-f)|Y|\leq \sum_{v\in Y}c(v)\leq K+2c(F,\bar{F})-2f(n-f)+f|Y|,
\end{eqnarray*}
and then 
\begin{eqnarray}
|Y|\leq \frac{K+2c(F,\bar{F})-2f(n-f)}{n-t-2f}
\end{eqnarray}
by the assumption $t+2f< n$. Since $c(v)\leq t+f$ for any $v\in F$, we have that $c(F,\bar{F})\leq f(t+f)$, and thus
\begin{eqnarray}
|Y|&\leq& \frac{K+2f(t+f)-2f(n-f)}{n-t-2f} \nonumber\\
&=&\frac{K+2f(t-n+2f)}{n-t-2f}
=\frac{K}{n-t-2f}-2f.
\end{eqnarray}
Finally,
\begin{eqnarray}
|U|&=&f+|Y|\leq \frac{K}{n-t-2f}-f. \label{ub}
\end{eqnarray}
\qed\end{proof}

{\flushleft \it Proof of Lemma \ref{lem:red} }:
First we show the time complexity.
The initial conflict numbers for all vertices can be computed in $O(n^2)$ time.
The while-loop is executed at most $n$ times, and each loop takes $O(n)$ time for finding and flipping a vertex, as well as updating the conflict numbers and the set $U$. The total time complexity is therefore $O(n^2)$.

The conclusions (i) and (ii) are trivial from the reduction algorithm, and (iii) follows from the correctnesses of the reduction rules, which are shown in Lemmas~\ref{red-rule} and \ref{lem:ub}. Note that, by rules R1 and R2, we only need to find the flipping set in $U$.
\qed

\section{Minimum total conflict number}

In this section we show an algorithm for {\sc Min-Sum 2-Clustering}, in which the cost function is defined by $h_1(\pi)=\sum_v c_\pi(v)$, i.e., the total conflict number. 
First we show how to cope with the simple case which will be excluded in the main procedure. 
A 2-partition $(V_1,V_2)$ is \emph{trivial} if $V_1$ or $V_2$ is empty; and is \emph{extreme} if $|V_1|=1$ or $|V_2|=1$.
\begin{lemma}\label{extreme}
Finding an extreme 2-partition $\pi$ with minimum $h_1(\pi)$ can be done in $O(n^2)$ time. 
\end{lemma}
\begin{proof}
The conflict number of any vertex $v$ in the 2-partition $(V,\emptyset)$ is $n-|N_G[v]|$. By Corollary~\ref{flipone}, if $v$ is a vertex with minimum $|N_G[v]|$, then $(\{v\},V\setminus \{v\})$ is an extreme 2-partition with minimum cost $h_1$. That is, we only need to find a vertex with minimum degree in $G$.
\qed\end{proof}

\begin{lemma}\label{lem:ms}
If $\pi$ is a 2-partition such that $h_1(\pi)\leq h_1(\pi\ominus v)$ for any vertex $v$, then $c_\pi(v)\leq (n-1)/2$ for each $v$.
\end{lemma}
\begin{proof}
If $c_\pi(v)>(n-1)/2$, then by Corollary~\ref{flipone} the profit of flipping $v$ is $4c_\pi(v)-2(n-1)>0$, and thus flipping $v$ decreases the $h_1$ cost.
\qed\end{proof}

\begin{lemma}\label{ms-flip}
Suppose that the $h_1$ cost of any extreme 2-partition is larger than $k$. 
If there exists a non-trivial and non-extreme 2-partition $\pi$ with $h_1(\pi)\leq k$, then there exists a $(k,t)$-feasible flipping set $F$ for $\pi_0$ of size at most $f$, where $\pi_0=(N_G[s],V\setminus N_G[s])$ for some vertex $s$, $t=(n-1)/2$, and $f=k/n$.
\end{lemma}
\begin{proof}
Let $\pi$ be a non-trivial and non-extreme 2-partition with minimum $h_1(\pi)$. Thus, $h_1(\pi)\leq k$ and both clusters of $\pi$ contain at least two vertices. 
By the minimality of $\pi$, if $h_1(\pi\ominus v)<h_1(\pi)$ for some $v$, then $\pi\ominus v$ must be extreme. However, it contradicts to the assumption that the $h_1$ cost of any extreme 2-partition is larger than $k$. 
Therefore $h_1(\pi)\leq h_1(\pi\ominus v)$ for any $v$, and by Lemma~\ref{lem:ms} we have that $c_{\pi}(v)\leq (n-1)/2$ for each $v$.

By the pigeonhole principle, there exists a vertex $s$ with $c_{\pi}(s)\leq h_1(\pi)/n\leq k/n$. Let $F=C_{\pi}(s)$. By definition, no vertex is in conflict with $s$ in $\pi\ominus F$, i.e., $\pi\ominus F=(N_G[s],V\setminus N_G[s])\equiv \pi_0$ which is the unique 2-partition with no conflict incident to $s$.
Therefore, $|F|\leq k/n$, and flipping $F$ in $\pi_0$ yields a 2-partition $\pi$ such that $h_1(\pi)\leq k$ and $c_{\pi}(v)\leq (n-1)/2$ for each $v$. That is, $F$ is
a $(k,t)$-feasible flipping set for $\pi_0$ of size at most $k/n$.
\qed\end{proof}

By Lemma~\ref{lem:red}, we can use the reduction algorithm with $K=k$, $t=(n-1)/2$ and $f=k/n$. The assumption $t+2f<n$ in Lemma~\ref{lem:red} is satisfied when $k\leq n^2/4$.
In fact, the largest value of $k$ we will use in this section is $0.185n^2$ (Corollary~\ref{ms:large}). 
The reduction algorithm returns $( U,\pi, m)$, and the remaining work is to search a flipping set $F\subseteq U$ with $h_1(\pi\ominus F)\leq k$ and $|F|\leq m$.
This work can be done by a simple search-tree algorithm which picks an arbitrary undetermined vertex $v$ and recursively solves the problem for two cases: flipping $v$ or not. Algorithm~\ref{alg:ms} is the proposed algorithm. Note that we need to try every vertex $s$ as the one in Lemma~\ref{ms-flip}.
\begin{algorithm}[t]
\caption{: Min-sum 2-clustering}\label{alg:ms}
{\bf Input:} a graph $G=(V,E)$ and integer $k$.\\
{\bf Output:} determining if existing $\pi$ with $h_1(\pi)\leq k$.
\begin{algorithmic}[1]
\If{existing an extreme 2-partition with cost at most $k$}
\State \Return True;
\EndIf 
\For{each $s\in V$} \label{v-loop}
\State $\pi_0\leftarrow (N_G[s],V\setminus N_G[s])$;
\State call {\sc Reduction}$(G,\pi_0,k,(n-1)/2,k/n)$ to compute $(U,\pi,m)$; \label{step:red}
\State $\chi\leftarrow c_\pi(X,X)$, where $X=V\setminus U$;\Comment{Preparing for tree-search}
\State construct the list $\mathcal{L}_1$ of $C_\pi(v,U)$ and the list $L_2$ of $c_\pi(v,X)$ for each $v\in U$;
\If{{\sc Search1}($U,m,L_2,\chi$)=True}
\State \Return True;
\EndIf
\EndFor
\State \Return False.
\end{algorithmic}
\end{algorithm}

A naive implementation of the search-tree algorithm takes $O(n^2)$ time for each recursive call, and the time complexity of the search-tree algorithm will be $O(n^2)$ multiplied by the number of recursive calls.
Similar to the technique usually used in the design of fixed-parameter algorithms \cite{nei00}, if the time complexity of each recursive call is a function in $|U|$ but not in $n$, then we can reduce the polynomial factor in the total time complexity, which is exactly the case shown in Lemma~\ref{time:search}. To this aim, the recursive procedure is designed in Algorithm~\ref{alg:search-ms}. 

\begin{algorithm}[t]
\caption{Search-tree algorithm for {\sc Min-Sum 2-Clustering}}\label{alg:search-ms}
{\bf Input:} a set $U$ of undetermined vertices, a flipping quota $m$, a list $L_2$ of $c(v,X)$ for each $v\in U$, and $\chi=c(X,X)$, where $X=V\setminus U$. In addition, a list $\mathcal{L}_1$ of $C(v,U)$ for each $v\in U$ is stored as a global variable.
\begin{algorithmic}[1]
\Procedure{Search1}{$U,m,L_2,\chi$}
\State {\bf if} $\chi>k$ {\bf then} \Return False; \label{step:sol1}
\If{$U=\emptyset$ or $m=0$}\Comment{no vertex can be flipped}
\State $q\leftarrow \chi+\sum_{v\in U}(|C(v,U)|+2c(v,X))$;\Comment{total conflict number}
\State {\bf if} $q\leq k$ {\bf then} \Return True {\bf else} \Return False;
\EndIf\label{step:sol2}
\State pick an arbitrary vertex $u\in U$;
\State $U'\leftarrow U\setminus \{u\}$; \Comment{$X'=V\setminus U'$}
\State modify $\mathcal{L}_1$: $C(v,U')\leftarrow C(v,U)\setminus \{u\}, \forall v\in U'$; record the modifications in $L_3$;
\State $\chi'\leftarrow \chi+2c_\pi(u,X)$;\Comment{move $u$ to $X$ without flipping}
\State construct $L_2'$ from $L_2$ by $c(v,X')\leftarrow c(v,X)+c(v,u), \forall v\in U'$;
\If{{\sc Search1}($U',m,L_2',\chi'$)=True}
\State \Return True;
\EndIf
\State $\chi''\leftarrow \chi+2(|X|-c(u,X))$;\Comment{flip and move $u$ to $X$}\label{step:flip1}
\State construct $L_2''$ from $L_2$ by $c(v,X')\leftarrow c(v,X)+1-c(v,u), \forall v\in U'$;\label{step:flip2}
\If{{\sc Search1}($U',m-1,L_2'',\chi''$)=True}
\State \Return True;
\EndIf
\State recover $\mathcal{L}_1$ by undoing the modifications in $L_3$;
\State \Return False;
\EndProcedure
\end{algorithmic}
\end{algorithm}
\begin{lemma}
The procedure {\sc Search1} is correct and each recursive call takes $O(|U|)$ time.
\end{lemma}
\begin{proof}
The correctness of the algorithm follows from the following three simple observations. First, it explores all subsets of $U$ with size at most $m$ in the worst case. Second,  
the total conflict number $h_1(\pi)$ can be computed as $c_\pi(X,X)+2c_\pi(U,X)+c_\pi(U,U)$, where $X=V\setminus U$. Third, when the vertex $u$ is flipped, the two sets of vertices conflicting and non-conflicting with $u$ exchange, and therefore the formulas at steps~\ref{step:flip1} and \ref{step:flip2} of Algorithm~\ref{alg:search-ms} are correct.

At each recursive call, there are only $O(|U|)$ data to be updated. 
The number of conflicting pairs in $U$ may be up to $\Theta(|U|^2)$. To avoid copying the conflicting pairs, the list $\mathcal{L}_1$ is stored as a global variable and the modifications are stored in a local variable $L_3$.
When returning from the recursive call with ``False'', $\mathcal{L}_1$ is recovered.
Note that it is not necessary to recover $\mathcal{L}_1$ when the recursive call returns ``True''. 
Since the number of modifications stored in $L_3$ is upper bounded by $O(|U|)$, the total time complexity for each recursive call is $O(|U|)$.
\qed\end{proof}

The algorithm we show here only returns True or False. In the case that a desired 2-partition needs to be output, we can record the flipped vertex at each recursive call (in constant time), and the 2-partition can be found by back tracking on the search tree in an additional $O(n)$ time.
Another thing that should be remarked is how to exclude trivial 2-partitions in the search-tree algorithm, which by definition are invalid. 
Let $\pi=(V_1,V_2)$ be the 2-partition returned by the reduction algorithm at step~\ref{step:red} of Algorithm~\ref{alg:ms}.
Recall that the initial 2-partition is $(N_G[s],V\setminus N_G[s])$ for some vertex $s$. Since the initial conflict number of $s$ is zero, the reduction algorithm never flips $s$, and therefore $s\in V_1\setminus U$.
Thus, $V_2$ is the only possible flipping set to result in a trivial 2-partition.
A simple way to avoid returning a trivial 2-partition uses a boolean variable which indicates whether there is a vertex fixed in $V_2$. When reaching a leaf of the branching tree, it can be easily verified if it is the invalid flipping set. 

\begin{lemma}\label{time:search}
Algorithm \ref{alg:search-ms} runs in $O(\phi^{|U|+m})$ time, where $\phi=\frac{1+\sqrt{5}}{2}\approx 1.618$.
\end{lemma}
\begin{proof}
Let $T(a,b)$ denote the time complexity of the search-tree algorithm with $|U|=a$ and $m=b$. There are two branches at each non-leaf node of the search-tree.
For the branch that a vertex is removed from $U$ without flipping, $|U|$ is decreased by one and $m$ is unchanged. For the branch that a vertex is flipped and removed, both $|U|$ and $m$ are decreased by one. Therefore, for some constant $p_1$,
$T(a,b)\leq T(a-1,b)+T(a-1,b-1)+p_1 a$ for $a,b>0$; $T(a,0)\leq p_1 a$ for any $a$; and $T(0,b)\leq p_1$. We shall show by induction that
\begin{eqnarray}\label{treeub}
T(a,b)\leq p_2\phi^{a+b}-p_1 (a+2)
\end{eqnarray}
for some constant $p_2$. Then, the time complexity is
$T(|U|,m)\in O(\phi^{|U|+m})$.

It is easy to see that, for sufficiently large $p_2$,
$T(0,b)\leq p_1 \leq p_2 \phi^{b}-2p_1$; and
$T(a,0)\leq p_1 a\leq p_2 \phi^{a}-p_1 (a+2)$.
Suppose by induction hypothesis that (\ref{treeub}) holds for $T(a-1,b-1)$ and $T(a-1,b)$. Note that $\phi$ is the solution of Fibonacci recursion and therefore $\phi^{i+2}=\phi^{i+1}+\phi^{i}$. For $a,b>0$,
\begin{eqnarray*}
T(a,b)&\leq& T(a-1,b)+T(a-1,b-1)+p_1 a \\
&\leq &(p_2\phi^{a+b-1}-p_1 (a+1))+(p_2\phi^{a+b-2}-p_1 (a+1))+p_1 a\\
&=&p_2\phi^{a+b}-p_1(a+2).
\end{eqnarray*}
\qed\end{proof}

\begin{theorem}\label{thm:minsum}
For $k\leq n^2/4$, {\sc Min-Sum 2-Clustering} can be solved in $O(n\cdot 2.619^{r/(1-4r/n)}+n^3)$ time, where $r=k/n$.
\end{theorem}
\begin{proof}
A non-trivial 2-partition is either extreme or non-extreme. Algorithm~\ref{alg:ms} copes with the case of extreme 2-partitions at step 1 which takes $O(n^2)$ time by Lemma~\ref{extreme}. 
If there is no extreme 2-partition with $h_1$ cost at most $k$, then Lemma~\ref{ms-flip} can be applied, and the non-extreme case is coped by the remaining steps.  
By (\ref{ub}), the number of undetermined vertices after reductions is  
\begin{eqnarray}
|U|\leq \frac{K}{n-t-2m}-m= \frac{k}{(n+1)/2-2m}-m\leq \frac{k}{n/2-2m}-m, \label{uub0}
\end{eqnarray}
where $K=k$ is the required bound of the total conflict number, $t=(n-1)/2$, and $m$ is the returned flipping quota.  
Since $m\leq k/n$, $|U|+m\leq k/(n/2-2m)\leq 2r/(1-4r/n)$.
By Lemma~\ref{time:search}, the time complexity of the search-tree algorithm is \begin{eqnarray}
O(\phi^{|U|+m})=O(\phi^{2r/(1-4r/n)})\subset O(2.619^{r/(1-4r/n)}) \label{uub1}
\end{eqnarray}
since $\phi^2=1+\phi<2.619$.
By Lemma~\ref{lem:red}, the reduction algorithm takes $O(n^2)$ time, and then the total time complexity follows from that the for-loop in Algorithm~\ref{alg:ms} is executed $n$ times.
\qed\end{proof}
For $k\in \Theta(n^2)$, the time complexity can be expressed as follows, in which the condition $\delta\leq 0.185$ is to ensure the result is better than the naive $O^*(2^n)$-time algorithm.
\begin{corollary}\label{ms:large}
For $k=\delta n^2$ with $\delta\leq 0.185$, {\sc Min-Sum 2-Clustering} can be solved in $O(n\cdot 2.619^{\delta n/(1-4\delta)}+n^3)$ time.
\end{corollary}

When $k\in o(n^2)$ and $n$ is sufficiently large, we have that $r=k/n\in o(n)$, and then  $1/(1-4r/n)<1+\varepsilon$ for any constant $\varepsilon>0$. The next corollary directly follows from Theorem~\ref{thm:minsum}.

\begin{corollary}\label{smallk}
For $k\in o(n^2)$, {\sc Min-Sum 2-Clustering} can be solved in $O(n\cdot 2.619^{k/n}+n^3)$ time.
\end{corollary}
If $k\in o(n^2)$, then $k/n\in o(n)$. Thus, by Corollary~\ref{smallk}, {\sc Min-Sum 2-Clustering} can be solved in $O^*(2^{o(n)})$ time, that is, in subexponential time.

\begin{corollary}
{\sc Min-Sum 2-Clustering} can be solved in polynomial time for $k\in O(n\log n)$.
\end{corollary}

\section{Minimizing the sum of squares}

Recall that $h_2(\pi)=\sum_v c^2_\pi(v)$ is the sum of squared conflict-numbers for a 2-partition $\pi$.
Given a graph $G$ and an integer $k$, {\sc Min-Square 2-Clustering} determines if there exists a 2-partition $\pi$ with $h_2(\pi)\leq k$. In this section, we show a parameterized algorithm for parameter $k$.

\begin{lemma}\label{mq1}
If $h_2(\pi)\leq k$, then $\sum_v c_\pi(v)\leq \sqrt{nk}$ and there exists a vertex $s$ with $c_\pi(s)\leq \sqrt{k/n}$.
\end{lemma}
\begin{proof}
By Cauchy-Schwarz inequality,
\[ \left(\sum_v c_\pi(v)\right)^2\leq \left(\sum_{i=1}^n 1^2\right)\left(\sum_v c^2_\pi(v)\right)=nh_2(\pi)\leq nk,\]
and we have that $\sum_v c_\pi(v)\leq \sqrt{nk}$.
The second consequence follows from $\min_v \{c_\pi(v)\}\leq (1/n)\sum_v c_\pi(v)$.
\qed\end{proof}

\begin{lemma}\label{mq2}
If $\pi$ is a 2-partition such that $h_2(\pi)\leq h_2(\pi\ominus v)$ for any vertex $v$, then $c_\pi(v)\leq \sqrt{n(n-1)/2}$ for each vertex $v$.
\end{lemma}
\begin{proof}
Consider $\pi'=\pi\ominus v$ for any vertex $v$.
First, $c_{\pi'}(v)=n-1-c_\pi(v)$.
Let $Y(v)=C_\pi(v)$ and $\bar{Y}(v)=V\setminus C_\pi(v)\setminus \{v\}$.
For each $u\in Y(v)$,
$c_{\pi'}(u)=c_\pi(u)-1$; and, for $u\in \bar{Y}(v)$, $c_{\pi'}(u)=c_\pi(u)+1$.
Therefore,
\begin{eqnarray*}
&&h_2(\pi)-h_2(\pi')\\
&=&c_{\pi}^2(v)-(n-1-c_\pi(v))^2 +\sum_{u\in Y(v)}\left(c_\pi^2(u)-(c_\pi(u)-1)^2\right) \\
&&+\sum_{u\in \bar{Y}(v)}\left(c_\pi^2(u)-(c_\pi(u)+1)^2\right)  \\
&=&2(n-1)c_\pi(v)-(n-1)^2+\sum_{u\in Y(v)}(2c_\pi(u)-1)-\sum_{u\in \bar{Y}(v)}(2c_\pi(u)+1)\\
&=&2(n-1)c_\pi(v)-n(n-1)+2\sum_{u\in Y(v)}c_\pi(u)-2\sum_{u\in \bar{Y}(v)}c_\pi(u).
\end{eqnarray*}
Since by the assumption $h_2(\pi)-h_2(\pi')\leq 0$, we have that
\begin{eqnarray}
c_\pi(v)\leq \frac{n}{2}+\frac{1}{n-1}\left(\sum_{u\in \bar{Y}(v)}c_\pi(u)-\sum_{u\in Y(v)}c_\pi(u)\right). \label{sos}
\end{eqnarray}
Let $s=\arg\max_v \{c_\pi(v)\}$.
By (\ref{sos}),
\begin{eqnarray*}
c_\pi(s)&\leq& \frac{n}{2}+\frac{1}{n-1}\left(\sum_{u\in \bar{Y}(s)}c_\pi(u)-\sum_{u\in Y(s)}c_\pi(u)\right) \\
&\leq& \frac{n}{2}+\frac{1}{n-1}\left(\sum_{u\in \bar{Y}(s)}c_\pi(u)\right) \\
&\leq& \frac{n}{2}+\frac{|\bar{Y}(s)|\cdot c_\pi(s)}{n-1} \\
&=&\frac{n}{2}+\frac{(n-1-c_\pi(s))c_\pi(s)}{n-1}
=\frac{n}{2}+c_\pi(s)-\frac{c_\pi^2(s)}{n-1}.
\end{eqnarray*}
That is, $c_\pi^2(s)\leq n(n-1)/2$, and we obtain
\begin{eqnarray}
c_\pi(s)\leq \sqrt{\frac{n(n-1)}{2}}.
\end{eqnarray}
\qed\end{proof}

Similar to Lemma~\ref{ms-flip}, we have the next corollary from Lemmas~\ref{mq1} and \ref{mq2}. Recall that a 2-partition is extreme if one of the two clusters is singleton. 
\begin{corollary}\label{mq-flip}
Suppose that the $h_2$ cost of any extreme 2-partition is larger than $k$. 
If there exists a non-trivial and non-extreme 2-partition $\pi$ with $h_2(\pi)\leq k$, then there exists a $(K,t)$-feasible flipping set $F$ for $\pi_0$ of size at most $f$, where  $\pi_0=(N_G[s],V\setminus N_G[s])$ for some vertex $s$, $K=\sqrt{nk}$, $t=\sqrt{n(n-1)/2}$, and $f=\sqrt{k/n}$.
\end{corollary}

The main steps of the algorithm for {\sc Min-Square 2-Clustering} are quite similar to Algorithm~\ref{alg:ms} in the previous section.
First, an extreme 2-partition with minimum cost $h_2$ can be found in $O(n^3)$ time since there are only $n$ extreme 2-partitions. Then, we can focus on non-extreme 2-partitions.  
By Corollary~\ref{mq-flip}, we use the reduction algorithm with the bound of total conflict number $K=\sqrt{nk}$, individual bound $t=\sqrt{n(n-1)/2}$ and flipping quota $f=\sqrt{k/n}$. The assumption $t+2f<n$ in Lemma~\ref{lem:red} is satisfied when $k\leq 0.021n^3$, 
and the largest value of $k$ we will use in this section is $0.0118n^3$ (Corollary~\ref{msq:large}).
For the reduced instance $U$ and $m$, a search-tree algorithm is employed to check if any desired 2-partition can be resulted from flipping a subset of $U$ with size at most $m$.
The recursive procedure is shown in Algorithm~\ref{alg:search2}. 
Unlike Algorithm~\ref{alg:search-ms}, we did not find a way to compute the cost $h_2$ with time complexity only depending on $|U|$ but not on $n$. Therefore it takes  $O(n^2)$ time for each recursive call, and the time complexity of the search-tree algorithm is $O(n^2)$ multiplied by the number of recursive calls. 
\begin{algorithm}[t]
\caption{Search-tree algorithm for {\sc Min-Square 2-Clustering}}\label{alg:search2}
{\bf Input:} a set $U$ of undetermined vertices, a flipping quota $m$, and a 2-partition $\pi$.
\begin{algorithmic}
\Procedure{{\sc Search2}}{$U,m,\pi$}
\If{$U=\emptyset$ or $m=0$}\Comment{no vertex can be flipped}
\State compute $h_2(\pi)$ and \Return True or False accordingly;
\EndIf
\State pick an arbitrary vertex $u\in U$;
\If{{\sc Search2}($U\setminus \{u\},m,\pi$)=True}\Comment{move $u$ to $X$ without flipping}
\State \Return True;
\ElsIf{{\sc Search2}($U\setminus \{u\},m-1,\pi\ominus u$)=True}\Comment{flip and move $u$ to $X$}
\State \Return True;
\Else
\State \Return False;
\EndIf
\EndProcedure
\end{algorithmic}
\end{algorithm}

\begin{theorem}\label{thm:sq}
For $k\leq 0.021n^3$, {\sc Min-Square 2-Clustering} can be solved in $O(n^3\cdot 5.171^{\theta/(1-(4+2\sqrt{2})\theta/n)})$, where $\theta=\sqrt{k/n}$.
\end{theorem}
\begin{proof}
By the reduction algorithm,
\begin{eqnarray}
|U|&\leq& \frac{K}{n-t-2m}-m \nonumber\\
&=&\frac{\sqrt{nk}}{n-\sqrt{n(n-1)/2}-2m}-m \nonumber\\
&\leq&\frac{\sqrt{nk}}{(1-1/\sqrt{2})n-2m}-m \nonumber\\
&=&\frac{(2+\sqrt{2})\theta}{1-(4+2\sqrt{2})m/n}-m. \label{ub:sq}
\end{eqnarray}

Let $T_0(a,b)$ denote the number of leaf nodes in the search tree explored by the algorithm {\sc Search2} with $|U|=a$ and $m=b$. According to the branching rule,
\begin{eqnarray}
T_0(a,b)\leq \left\{\begin{array}{ll}
T_0(a-1,b)+T_0(a-1,b-1) &\mbox{if }a,b>0; \\
1 &\mbox{otherwise.}
\end{array}\right.\label{t0}
\end{eqnarray}
Next we show by induction that
\begin{eqnarray}
T_0(a,b)\leq \phi^{a+b}. \label{t0b}
\end{eqnarray}
It is clear that (\ref{t0b}) holds when $a=0$ or $b=0$. Suppose by the induction hypothesis that it holds for $T_0(a-1,b)$ and $T_0(a-1,b-1)$. Then, for $a,b>0$,
\[ T_0(a,b)=T_0(a-1,b)+T_0(a-1,b-1)\leq \phi^{a+b-1}+\phi^{a+b-2}= \phi^{a+b} \]
by the identity $\phi^2=\phi+1$.
Since the search tree is binary, the number of recursive calls is bounded by $2T_0(|U|,m)\in O(\phi^{|U|+m})$. 
By (\ref{ub:sq}), since $m\leq \sqrt{k/n}=\theta$,
\begin{eqnarray*}
|U|+m&\leq& \frac{(2+\sqrt{2})\theta}{1-(4+2\sqrt{2})m/n}\\
&\leq&\frac{(2+\sqrt{2})\theta}{1-(4+2\sqrt{2})\theta/n}.
\end{eqnarray*}

Similarly to Algorithm~\ref{alg:ms}, the reduction and the search-tree algorithms are executed $n$ times. 
Since each recursive call take $O(n^2)$ time, the total time complexity is
$O(n^3\cdot \phi^{|U|+m})\subset O(n^3\cdot 5.171^{\theta/(1-(4+2\sqrt{2})\theta/n)})$.
\qed\end{proof}
For $k\in \Theta(n^3)$, the time complexity can be expressed as follows, in which the condition $k\leq 0.0118n^3$ is to ensure the result is better than the naive $O^*(2^n)$-time algorithm.
\begin{corollary}\label{msq:large}
For $k=\delta^2n^3$ with $\delta^2\leq 0.0118$, {\sc Min-Square 2-Clustering} can be solved in $O(n^3\cdot 5.171^{\delta n/(1-(4+2\sqrt{2})\delta)})$.
\end{corollary}
When $k\in o(n^3)$ and $n$ is sufficiently large, we have that $\theta=\sqrt{k/n}\in o(n)$, and then  $1/(1-(4+2\sqrt{2})\theta/n)<1+\varepsilon$ for any constant $\varepsilon>0$. 
The next corollary directly follows from Theorem~\ref{thm:sq}, which implies that {\sc Min-Square 2-Clustering} can be solved in subexponential time for $k\in o(n^3)$.

\begin{corollary}
For $k\in o(n^3)$, {\sc Min-Square 2-Clustering} can be solved in $O(n^3\cdot 5.171^{\theta})$ time, where $\theta=\sqrt{k/n}$.
\end{corollary}

\begin{corollary}
For $k\in O(n\log^2 n)$, {\sc Min-Square 2-Clustering} can be solved in polynomial time.
\end{corollary}

\section{Concluding remarks}

In this paper, we show parameterized algorithms for {\sc Min-Sum 2-Clustering} and {\sc Min-Square 2-Clustering}. The first problem is the same as {\sc 2-Cluster Editing} in the literature with an additional multiplicative factor two in the cost function.
The proposed algorithms run in subexponential time and significantly improve the
brute-force algorithm when $k$ is relatively small, i.e., $k\in o(n^2)$ for {\sc Min-Sum 2-Clustering} and $k\in o(n^3)$ for {\sc Min-Square 2-Clustering}.

The time complexities of the search-tree algorithms are shown by induction in Lemma~\ref{time:search} and Theorem~\ref{thm:sq}. It can be shown that the time complexities of both algorithms are almost tight when $k$ is relatively small. We shall show the case of Algorithm~\ref{alg:search-ms} with $k\in o(n^2)$, and it is similar for Algorithm~\ref{alg:search2} with $k\in o(n^3)$. 
In worst case, the search-tree algorithm explores all subsets of $U$ with sizes at most $m$.
The number of recursive calls is lower bounded by $\sum_{i=0}^m {|U|\choose i}$. From (\ref{uub0}), when $k\in o(n^2)$ and $m=k/n$, we have that $m\in \Theta(|U|)$. 
Let $m=\alpha |U|$ for some constant $\alpha$.
By \cite[Lemma 3.13]{expbook}, 
\[ \sum_{i=0}^m{|U|\choose i}\geq p|U|^{-1/2}\cdot 2^{H(\alpha)|U|}, \]
where $p$ is a constant and $H(\alpha)=-\alpha\log_2 \alpha-(1-\alpha)\log_2(1-\alpha)$ is the \emph{binary entropy function}. 
When $\alpha=\phi^{-2}=(3-\sqrt{5})/2$, it can be verified that $2^{H(\alpha)|U|}=\phi^{|U|+m}$. 

The lower bounds for the two problems are interesting open problems.
Another straightforward question is how to generalize the algorithms to the case of more than two clusters. However, the problem seems to become much more difficult when the number of clusters is more than two. In fact, by the following simple transformation, we can show that there is no algorithm solving {\sc 3-Cluster Editing} in $O^*(2^{k/n})$ time unless NP=P, where $k$ and $n$ are the cost bound and the number of vertices. Let $(G',k)$ be an instance of the NP-complete {\sc 2-Cluster Editing} problem. We construct $G$ from $G'$ by adding an isolated clique of $k+2$ vertices. One can observe that $G$ can be edited into 3 clusters with cost $k$ if and only if $G'$ can be edited into 2 clusters with the same cost.
Since $n>k$, an algorithm solving {\sc 3-Cluster Editing} in $O^*(2^{k/n})$ time can also solve the NP-complete {\sc 2-Cluster Editing} problem 
in polynomial time.

\section*{acknowledgements}
The authors would like to thank the anonymous referees for their helpful comments which improved the presentation significantly. 
This work was supported in part by 
NSC 100-2221-E-194-036-MY3 and NSC 101-2221-E-194-025-MY3 from the National Science Council, Taiwan.

\bibliographystyle{spmpsci}      
\bibliography{2clusterb}

\end{document}